\documentclass[11pt,letterpaper]{article}

\usepackage{libertine}
\usepackage{fullpage}
\usepackage{booktabs} % For formal tables
\usepackage{amsmath,amsfonts,amssymb,bbm,amsthm} 

\usepackage{graphicx}
\usepackage{subcaption}
\usepackage{verbatim}
\usepackage{url}
\usepackage{multirow}
\usepackage{enumitem}

\usepackage{hyperref}
\usepackage[svgnames]{xcolor}
\usepackage[capitalise,nameinlink]{cleveref}
\hypersetup{colorlinks={true},linkcolor={DarkBlue},citecolor=[named]{DarkGreen}}

\usepackage{natbib}

\newcommand{\R}{\mathbb{R}}
\newcommand{\X}{\mathbb{X}}

\newcommand{\xx}{\mathbf{x}}
\newcommand{\sss}{\mathbf{s}}
\newcommand{\yy}{\mathbf{y}}

\newcommand{\dd}{\mathbf{d}}

\newcommand{\SW}{\text{\normalfont SW}}
\newcommand{\LW}{\text{\normalfont LW}}

\newcommand{\lpoa}{\text{\normalfont LPoA}}
\newcommand{\poa}{\text{\normalfont PoA}}

\newcommand{\calG}{\mathcal{G}}
\newcommand{\calQ}{\mathcal{Q}}
\newcommand{\calC}{\mathcal{C}}
\newcommand{\ud}{\,\mathrm{d}}

\newcommand{\zero}{\mathbf{0}}

\newcommand{\eq}{\text{eq}}

\newcommand{\prop}{\text{Kelly}}
\newcommand{\SH}{\text{SH}}

\newcommand{\PYS}{\text{PYS}}

\newcommand{\EPYS}{\text{E2-PYS}}
\newcommand{\SR}{\text{E2-SR}}
\newcommand{\SHSR}{\text{SH-SR}}

\newtheorem{theorem}{Theorem}[section]
\newtheorem{lemma}[theorem]{Lemma}

\title{\bf The efficiency of resource allocation mechanisms\\for budget-constrained users}  

\author{Ioannis Caragiannis \\ Department of Computer Engineering and Informatics, University of Patras, Greece  
\and Alexandros A. Voudouris \\ Department of Computer Science, University of Oxford, UK}
\date{}

\begin{document}

\maketitle

\begin{abstract}
We study the efficiency of mechanisms for allocating a divisible resource. Given scalar signals submitted by all users, such a mechanism decides the fraction of the resource that each user will receive and a payment that will be collected from her. Users are self-interested and aim to maximize their utility (defined as their value for the resource fraction they receive minus their payment). Starting with the seminal work of Johari and Tsitsiklis [Mathematics of Operations Research, 2004], a long list of papers studied the price of anarchy (in terms of the social welfare --- the total users' value) of resource allocation mechanisms for a variety of allocation and payment rules. Here, we further assume that each user has a budget constraint that invalidates strategies that yield a payment that is higher than the user's budget. This subtle assumption, which is arguably more realistic, constitutes the traditional price of anarchy analysis meaningless as the set of equilibria may change drastically and their social welfare can be arbitrarily far from optimal. Instead, we study the price of anarchy using the liquid welfare benchmark that measures efficiency taking budget constraints into account. We show a tight bound of $2$ on the liquid price of anarchy of the well-known Kelly mechanism and prove that this result is essentially best possible among all multi-user resource allocation mechanisms. This comes in sharp contrast to the no-budget setting where there are mechanisms that considerably outperform Kelly in terms of social welfare and even achieve full efficiency. In our proofs, we exploit the particular structure of worst-case games and equilibria, which also allows us to design (nearly) optimal two-player mechanisms by solving simple differential equations.
\end{abstract}

\maketitle

\section{Introduction}\label{sec:intro}
Resource allocation is an ubiquitous task in computing systems and often sets algorithmic challenges to their design. As such, resource allocation problems have received much attention by the algorithmic community for decades. The recent emergence of large-scale distributed systems with non-cooperative users that compete for access to scarce resources has led to game-theoretic treatments of resource allocation.

In this paper, we study a particular simple class of {\em resource allocation mechanisms} that aim to distribute a divisible resource (such as bandwidth of a communication link, CPU time, storage space, etc.) by auctioning it off to different users as follows. Each user is asked to submit a scalar {\em signal}. Given the submitted signals, the mechanism decides the fraction of the resource that will be allocated to each user, as well as the payment that will be received from each of them. A typical example is a mechanism that has been proposed by \citet{K97} (henceforth called the Kelly mechanism; see also \citep{KMT98}), according to which the fraction of the resource allocated to each user is proportional to the user's signal, and the signal itself is her payment.

The users are self-interested. Following the standard modeling assumptions in the related literature, the value of each user for a resource fraction is given by a private {\em valuation function}. The above definition of resource allocation mechanisms allows the users to act strategically in the sense that the signal they select to submit is such that their utility (value for the fraction of the resource they receive minus payment) is maximized. Naturally, this behavior defines a {\em strategic game} among the users, who act as players. Soon after the definition of the Kelly mechanism, a series of papers studied the existence and uniqueness of {\em pure Nash equilibria} (snapshots of player strategies, in which the signal of each player maximizes her own utility) of the induced games \citep{HG02,LA00,MB03} and quantified their inefficiency \citep{JT04} using the notion of the {\em price of anarchy}~\citep{KP99}.

In particular, \citet{JT04} used the {\em social welfare} (i.e., the total value of the players for their received fraction of the resource) as an efficiency benchmark and proved that the social welfare at any equilibrium is at least $3/4$ times the optimal social welfare. This translates into a price of anarchy bound of $4/3$, which is tight. The paper of \citet{JT04} sparked subsequent research on other resource allocation mechanisms, that use different (non-proportional) allocation rules or payments. 

A first apparent question was whether improved price of anarchy bounds are possible by changing the allocation function, but keeping the simple pay-your-signal (or PYS, for short) payment rule. \citet{SH04} showed that no PYS mechanism has price of anarchy better than $8/7$, designed an allocation function that achieves this bound for two players, and provided strong experimental evidence that a slightly inferior bound holds for arbitrarily many players. Surprisingly, full efficiency at equilibria (i.e., a price of anarchy equal to $1$) is possible via different allocation/payment functions. This discovery was made in three independent papers by \citet{MB06}, \citet{YH07}, and \citet{JT09}. The mechanism of \citet{MB06} uses proportional allocation (but a different payment; see Section \ref{sec:prelim} for its description), while the mechanisms of \citet{JT09} and \citet{YH07} are adaptations of the well-known VCG paradigm (see also the survey by \citet{J07} on these results).

Our focus in the current paper is on the ---arguably, more realistic--- setting, in which each player has a private {\em budget} that constrains the payments that she can afford and, consequently, narrows her strategy space. As resource allocation mechanisms do not have direct access to budgets, budget constraints can restrict the set of equilibria so that their social welfare is extremely low compared to the optimal social welfare, which in turn is not related to player strategies, payments, or budgets. An efficiency benchmark that is suitable for budget-constrained players is known as {\em liquid welfare} (introduced by \citet{DPL14} and, independently, by \citet{ST13} who call it {\em effective} welfare) and is obtained by slightly changing the definition of the social welfare, taking budgets into account. Informally, the liquid welfare is the total value of the players for the resource fraction they receive, with the value of each player capped by her budget. 

\citet{DPL14} provide the following justification of liquid welfare as a natural extension of a revenue-motivated interpretation of social welfare. In the no-budget setting, besides simply viewing the social welfare as the total utility of all players and the mechanism combined, we can also interpret it as an upper bound on the revenue that the mechanism can hope to gain by ``selling'' fractions of the resource to the players. In the setting with budgets, the liquid welfare admits the same interpretation. If the value of a player exceeds her budget, then the maximum revenue is the budget (which restricts what she can afford). On the other hand, if the budget exceeds her value, then the value is the maximum revenue as if there was no budget constraint. Combining these together, we get that the maximum revenue is the minimum between these two, which gives rise to the definition of liquid welfare, by summing over all players. In addition, \citet{DPL14} compare liquid welfare to other benchmarks previously used in the  literature. 

Following the recent paper of \citet{AFGR15}, we use the term {\em liquid price of anarchy} (and abbreviate it as $\lpoa$) to refer to the price of anarchy with respect to the liquid welfare, i.e., the ratio between the optimal liquid welfare of a game induced by a resource allocation mechanism and the worst liquid welfare over all equilibria of the game.

\paragraph{Our results and techniques.}
We aim to explore {\em all} resource allocation mechanisms and find the mechanism with the best possible $\lpoa$. Our results suggest a drastically different picture compared to the no-budget setting. First, the analogue of full efficiency is not achievable; we show a lower bound of $2-1/n$ on the $\lpoa$ of any $n$-player resource allocation mechanism (under standard technical assumptions for player valuations and mechanism characteristics; see Section \ref{sec:prelim}). The Kelly mechanism is proved to have an almost best possible $\lpoa$ of exactly $2$. In contrast, the mechanism of \citet{SH04} (henceforth called $\SH$) has an $\lpoa$ of $3$. Improved bounds are possible for two players. We design the two-player PYS resource allocation mechanism $\EPYS$ that has an $\lpoa$ of $1.792$; this bound is optimal among a very broad class of mechanisms. We also design the two-player mechanism $\SR$ that achieves an almost optimal $\lpoa$ bound of at most $1.529$; this mechanism uses different payments. We also present a simpler $2$-player mechanism, which combines the allocation function of SH and the payment function of $\SR$, and achieves a tight $\lpoa$ bound of $\phi=1.618$. See Table \ref{tab:results} for a summary of our results.

\begin{table}[t]
\centering
\begin{tabular}{l c p{11cm}}
\noalign{\hrule height 1pt}\hline
Mechanism 	& LPoA  & Comment \\
\noalign{\hrule height 1pt}\hline
all     & $\geq 2-1/n$ & No mechanism can achieve full efficiency (Theorem \ref{thm:lower-bound}) \\
$\prop$ & $2$ &  Tight bound; almost optimal among all $n$-player mechanisms (Theorem \ref{thm:proportional-upper}) \\
$\SH$ 	& $3$ & Tight bound (Theorems \ref{thm:SH-upper} and \ref{thm:SH-lower}) \\
$\EPYS$ & $1.792$ &  Tight bound (Theorem \ref{thm:E2-PYS-bound}); optimal among all 2-player $\PYS$ mechanisms with concave allocation functions (Theorem \ref{thm:PYS-optimality}) \\
$\SR$   & $\leq 1.529$ &  Almost optimal among all 2-player mechanisms (Theorem \ref{thm:RS-bound}) \\
$\SHSR$ & $1.618$ &  Tight bound (Theorem \ref{thm:SHSR-bound}) \\
\noalign{\hrule height 1pt}\hline
\end{tabular}
\caption{Summary of our liquid price of anarchy bounds for resource allocation mechanisms for budget-constrained users. The last three mechanisms are specifically for two players. The term ``tight bound'' means that the analysis of the corresponding mechanism is tight.}
\label{tab:results}
\end{table}  

Our results exploit a particular structure of worst-case (in terms of $\lpoa$) games and their equilibria.  We prove that for every resource allocation mechanism, the worst-case $\lpoa$ is obtained at instances in which players have {\em affine} valuation functions. In addition, all players besides one have finite budgets and play strategies that imply payments that are either zero or equal to their budget, while a single player has infinite budget and a signal that nullifies the derivative of her utility. Compared to an analogous characterization for the no-budget case (with linear valuation functions and player signals that all nullify their utility derivatives), first observed by \citet{JT04} for the Kelly mechanism and later extended to all resource allocation mechanisms, the structure in our characterization is much richer and the proof is considerably more complicated. The characterization contains so much information that the $\lpoa$ bounds follow rather easily; the extreme example is the proof of our best $\lpoa$ bound of $2$ for the Kelly mechanism which is only a few lines long. It can also be used in the design of new mechanisms; for example, the design and analysis of our two-player mechanisms $\EPYS$ and $\SR$ follow by simple first-order differential equations, which would never have been identified without our characterization. And, furthermore, under assumptions about the resource allocation mechanisms (e.g., concave allocations and convex payments), the $\lpoa$ bound is automatically proved to be tight without providing any explicit lower bound instance.

\paragraph{Other related work.} As an efficiency benchmark, liquid welfare has been studied recently in different contexts such as in the design of truthful mechanisms (see \citep{DPL14,LX15,LX17}) and in the analysis of combinatorial Walrasian equilibria with budgets \citep{DEF+16}. In the context of the price of anarchy, it was considered recently in simultaneous first price auctions by \citet{AFGR15} and in position auctions by \citet{V18}.

\citet{CV16} were the first to prove that the liquid price of anarchy of the Kelly mechanism is constant. In particular, they showed $\lpoa$ upper and lower bounds of $2.78$ and $2$, respectively. The lower bound is essentially proved again here (see Theorem \ref{thm:proportional-upper}) with a completely different and more interesting technique. \citet{CST16} improved the $\lpoa$ upper bound to $2.618$ and extended the results to more general settings involving multiple resources. Prior to these two papers, \citet{ST13} proved that the social welfare at equilibria of the Kelly mechanism is at most a constant factor away from the optimal liquid welfare. In contrast to the analysis techniques in the current paper, the analysis of the Kelly mechanism by \citet{CV16}, \citet{CST16} and \citet{ST13} is closer in spirit to the smoothness template (see \citep{R15,RST16}) and is based on bounding the utility of each player by the utility she would have when deviating to appropriate signal strategies. Their results hold for more general equilibrium concepts such as coarse-correlated or Bayes-Nash equilibria. Our $\lpoa$ bounds in the current paper hold specifically for pure Nash equilibria, but are superior and tight.

\paragraph{Roadmap.} The rest of the paper is structured as follows. We begin with definitions and notation in Section \ref{sec:prelim}. Our lower bound on the $\lpoa$ of any resource allocation mechanism appears in Section \ref{sec:lower-bound}. Section~\ref{sec:structure} is devoted to proving the structural characterization of worst-case resource allocation games and equilibria. Then, in Section \ref{sec:pay-your-signal} we present tight bounds on the liquid price of anarchy for the Kelly and SH mechanisms. In Section \ref{sec:two-players}, we present our two-player mechanisms $\EPYS$, $\SR$ and $\SHSR$. We conclude with open problems and a discussion on possible extensions in Section \ref{sec:extensions}.

%%%%%%%%%%%%%%%%
%%%%%%%%%%%%%%%%

\section{Definitions and notation}\label{sec:prelim} 
We consider a single divisible resource of unit size that is distributed among $n$ users by a {\em resource allocation mechanism} $M$.
The mechanism $M$ consists of
\begin{itemize}
\item an {\em allocation function} $g^M:\R^n_{\geq 0} \rightarrow \calQ \cup \zero$, where $\calQ = \left\{ \dd \in [0,1]^n: \sum_{i=1}^n d_i =1 \right\}$ is the unit $n$-simplex and $\zero = (0, ..., 0)$, and
\item a {\em payment function} $p^M:\R^n_{\geq 0} \rightarrow \R_{\geq 0}^n$,
\end{itemize}
and works as follows. Each user $i$ submits a {\em signal} $s_i \in \R_{\geq 0}$, and the mechanism $M$ allocates a fraction of $g^M_i(\sss)$ of the resource to each user $i$ and asks her for a payment of $p^M_i(\sss)$, where $\sss=(s_1,...,s_n)$ denotes the vector formed by all signals. Clearly, $g_i^M$ and $p_i^M$ are the $i$-th entries of vectors $g^M$ and $p^M$, respectively.

Some important properties of allocation and payment functions are as follows:
\begin{itemize}
\item They are {\em anonymous}: any permutation of the entries of the input signal vector results in the same permutation of the output. So, all users get equal resource shares and are asked for equal payments when they submit identical signals;
\item The mechanism does not allocate any fraction and does not ask for any payment from a user that submits a zero signal;
\end{itemize}   
Let $(y,\sss_{-i})$ denote the signal vector in which user $i$ has signal $y$ and the remaining users have their signals as in $\sss$. Viewed as univariate functions (of variable $y$), the functions $g_i^M(y,\sss_{-i})$ and $p_i^M(y,\sss_{-i})$ are {\em increasing} (with the exception of $(y,\sss_{-i})=(y,\zero_{-i})$ where $g^M_i(y,\zero_{-i})=1$) and {\em differentiable} in $\R_{\geq0}$ (with the exception of $(y,\sss_{-i})=\zero$). The above are characteristics that all known mechanisms in the resource allocation literature have (e.g., see~\citep{J07,JT04,JT09,MB03,SH04,YH07}). 

Each user $i$ has 
\begin{itemize}
\item a monotone non-decreasing, concave, and differentiable {\em valuation function} $v_i:[0,1]\rightarrow \R_{\geq 0}$, \footnote{Note that we do not assume that the valuation functions are normalized and it might be the case that $v_i(0)\neq 0$ for some player $i$, which is quite uncommon in the related literature. However, observe that any concave function $v$ with an offset of $\alpha > 0$ (i.e., $v(0)=\alpha$) can be approximated by another concave function $\tilde{v}$ with no offset (i.e., $\tilde{v}(0)=0$) but slope large enough so that $\tilde{v}(\epsilon)=v(\epsilon)$ for some arbitrarily small $\epsilon > 0$; $x \geq \epsilon$, the two functions $v$ and $\tilde{v}$ have values that coincide.} so that $v_i(x)$ represents the value that user $i$ has for a resource fraction of $x$;

\item a {\em private budget} $c_i \in \R_{\geq 0}\cup\{ +\infty \}$, which restricts (upper-bounds) her payment to the mechanism. 
\end{itemize}
Her {\em utility} from the mechanism is defined as the value she gets for the fraction she is given minus her payment, i.e., 
$$u^M_i(\sss) = v_i\left(g_i^M(\sss)\right) - p_i^M(\sss).$$ 
To capture the fact that budgets impose hard constraints to the users, we technically assume that $u^M_i(\sss)=-\infty$ when $p_i^M(\sss)>c_i$. 

The users act strategically as utility maximizers and, therefore, engage as players into a strategic {\em resource allocation game} $\calG^M$ that is {\em induced} by mechanism $M$. A (pure Nash) {\em equilibrium} is a signal vector $\sss$ such that, when viewed as a univariate function of variable $y$, $u^M_i(y,\sss_{-i})$ is maximized for $y=s_i$, i.e., no player can increase her utility by unilaterally deviating to submitting a different signal. We denote by $\eq(\calG^M)$ the set of all equilibria of game $\calG^M$. 

Due to the budget constraints, we have three different cases for the strategy of player $i$ at an equilibrium $\sss\in \eq(\calG^M)$ (assuming a non-trivial budget $c_i>0$) and for the corresponding value of the derivative of her utility. In particular, the derivative $\left.\frac{\partial u^M_i(y,\sss_{-i})}{\partial y}\right|_{y=s_i}$ is equal to zero in case $s_i$ is such that $0<p^M_i(\sss)<c_i$,  non-positive in case $s_i=0$, and non-negative in case $s_i$ is such that $p^M_i(\sss)=c_i$. Note that nullification of the utility derivative does not necessarily imply maximization of utility.

We argue, using the definition and properties of the allocation and payment functions, that signal vectors with at most one positive entry cannot be equilibria. This is clearly the case for the signal vector $\zero$. Indeed, any player $i$ with positive budget and valuation function that takes positive values has the incentive to unilaterally deviate to submitting a sufficiently small signal, so that the mechanism gives her the whole resource and asks for a payment of, say, $v_i(1)/2$; this is due to the continuity of the payment function and the fact that it is equal to zero on input the signal vector $\zero$. Actually, player $i$ has an incentive to further decrease her payment. At some point, she will deviate to a signal $\delta>0$ so that the payment of another player $j$ when both $i$ and $j$ have a signal equal to $\delta$ is, say, strictly less than $v_j(1/2)$. At that point, player $j$ will also deviate to the positive signal $\delta$ and get half of the resource and positive utility.\footnote{In general, we consider games in which at least two players have strictly positive budgets and valuation functions that take non-zero values.} We use $\X_n$ as an abbreviation of the subset of $\R^n_{\geq 0}$ consisting of signal vectors with at least two strictly positive entries.

We are interested in studying the effect of the strategic behavior on the efficiency of resource allocation mechanisms. An efficiency benchmark that has been used extensively in the related literature is {\em social welfare}. For an allocation $\dd \in \calQ\cup \zero$ of a resource allocation game $\calG^M$, the social welfare is defined as
$$\SW(\dd,\calG^M)= \sum_{i=1}^n v_i(d_i),$$
where $n$ is the number of players in $\calG^M$ and $v_i$ is the valuation function of player $i$. Then, the inefficiency of equilibria of game $\calG^M$ can be measured by its {\em price of anarchy} which is defined as 
$$\poa(\calG^M) = \sup_{\sss \in \eq(\calG^M)} \frac{\SW^*(\calG^M)}{\SW(g^M(\sss),\calG^M)},$$
where $\SW^*(\calG^M)$ denotes the maximum social welfare over all allocations of $\calG^M$.

However, the definition of the social welfare does not take into account the possibly finite budgets that the players may have. Therefore, we instead use the {\em liquid welfare} as our efficiency benchmark. The liquid welfare of an allocation $\dd$ is defined as 
$$\LW(\dd,\calG^M)= \sum_{i=1}^n \min\{v_i(d_i),c_i\},$$
where $c_i$ is the budget of player $i$. Clearly, when players have no budget constraints, the liquid welfare coincides with the social welfare. The {\em liquid price of anarchy} of a resource allocation game $\calG^M$ is then defined as 
$$\lpoa(\calG^M) = \sup_{\sss \in \eq(\calG^M)} \frac{\LW^*(\calG^M)}{\LW(g^M(\sss),\calG^M)},$$
where $\LW^*(\calG^M)$ denotes the maximum liquid welfare over all allocations of game $\calG^M$. We use the overloaded term $\lpoa(M)$ to denote the liquid price of anarchy of the resource allocation mechanism $M$. This is defined as the maximum (or, more formally, the supremum) liquid price of anarchy over all games that are induced by mechanism $M$. 

\subsection{Examples of resource allocation mechanisms}
Let us devote some space to the definition of some well-known mechanisms from the literature. An important class of resource allocation mechanisms is that of {\em pay-your-signal} mechanisms ($\PYS$, for short). When at least two players submit non-zero signals, a $\PYS$ mechanism charges each player $i$ a payment equal to the signal $s_i$ that she submits. Otherwise, $\PYS$ mechanisms follow the general convention that we have defined at the beginning of Section~\ref{sec:prelim}, and do not charge any payment to any player.

The most popular $\PYS$ mechanism is the $\prop$ mechanism that was introduced in \citep{K97}. This mechanism allocates the resource proportionally to the players' signals (this is why it is also known as the proportional allocation mechanism in the related literature), i.e.,
$$g_i^{\prop}(\sss) = \frac{s_i}{\sum_{j=1}^n{s_j}}.$$ 

The $\prop$ mechanism has played a central role in the related literature; for the no-budget setting, \citet{JT04} proved that its price of anarchy is $4/3$. In their attempt to design the $\PYS$ mechanism with the lowest possible price of anarchy, \citet{SH04} defined the allocation function 
$$g_i^{\SH}(\sss) = \frac{s_i}{\max_\ell\{s_\ell\}} \int_0^1 \prod_{j \neq i} \left( 1- \frac{s_j}{\max_\ell\{s_\ell\}}t \right) \ud t.$$
We will refer to the $\PYS$ mechanism that uses this allocation function as $\SH$. For two players, the allocation function has a very simple definition as $g_1^{\SH}(\sss)=\frac{s_1}{2s_2}$ when $s_1\leq s_2$, and $g_1^{\SH}(\sss)=1-\frac{s_2}{2s_1}$ otherwise. \citet{SH04} proved that the two-player version of the $\SH$ mechanism has an optimal (among all $\PYS$ mechanisms) price of anarchy of $8/7$ and provided experimental evidence that the price of anarchy of the $n$-player version is only marginally higher. As we will see later in Section~\ref{sec:pay-your-signal}, the comparison between $\prop$ and $\SH$ yields a drastically different result when players have budgets and the liquid welfare is used as the efficiency benchmark.

Other interesting classes of mechanisms use proportional allocation, but different kinds of payments. Among them, a mechanism defined by \citet{MB06} uses the class of payment functions  
\begin{align*}
p_i^M(\sss) = \left(\sum_{j \neq i}{s_j}\right)\cdot \int_0^{s_i} \frac{h^M(t+\sum_{j \neq i}s_j)}{(t+\sum_{j \neq i}s_j)^2} \ud t,
\end{align*}
where $h^M:\R_{\geq 0} \rightarrow \R_{\geq 0}$ is an increasing function (such as $h^M(z)=z$; several other choices for $h^M$ have been suggested in \citep{MB06}). These mechanisms have the remarkable property of full efficiency at equilibria in the no-budget setting (i.e., they have price of anarchy equal to $1$). Independently from \citet{MB06}, \citet{JT09} as well as \citet{YH07} presented resource allocation mechanisms that achieve full efficiency in the no-budget setting. All these mechanisms can be thought of as adaptations of the well-known VCG paradigm.

\section{A lower bound for all mechanisms}\label{sec:lower-bound}
The fact that the mechanisms of \citep{MB06,JT09,YH07} achieve full efficiency seems quite surprising, since resource allocation mechanisms do not have direct access to the valuation functions of the players. The definition of these mechanisms is such that the incentives of the players are fully aligned to the global goal of maximizing the social welfare. In a sense, these mechanisms manage to achieve access to the valuation functions indirectly.
In contrast, when players have budget constraints, we show below that a liquid price of anarchy equal to $1$ is not possible. This means that resource allocation mechanisms fail to ``mine'' any kind of information about the budget values of the players, while budgets affect the strategic behavior of the players crucially.\footnote{One might wonder whether this inability of resource allocation mechanisms is due to the restriction that each player submits a single scalar signal to the mechanism. This is not true, as we discuss in Section~\ref{sec:extensions}. In particular, we show that the lower bound of $2-1/n$ holds even if the players submit multidimensional signals, which could be used to encode detailed information about their valuation functions and budgets. Moreover, in Section~\ref{sec:extensions}, we show a slightly weaker lower bound even for mechanisms that know the budgets of the players a priori.}

\begin{theorem}\label{thm:lower-bound}
Every $n$-player resource allocation mechanism has liquid price of anarchy at least $2-1/n$.
\end{theorem}

\begin{proof}
Let $M$ be any $n$-player resource allocation mechanism that uses an allocation function $g^M$ and a payment function $p^M$. Let $\sss=(s_1, ..., s_n)$ be an equilibrium of the game $\calG^M_1$ induced by $M$ for players with valuations $v_i(x)=x$ and budgets $c_i=+\infty$, for every $i \in [n]$. Assume that the allocation returned by $M$ at this equilibrium is $\dd = (d_1, ..., d_n)$. Since all players have the same linear valuation function and none of them has a budget constraint, any allocation (equilibrium or not) achieves the same liquid (or social) welfare, and hence $\lpoa(\calG^M_1) = 1$.

Recall that, for every signal vector $\yy = (y_1, ..., y_n)$, the utility of player $i$ is defined as $u^M_i(\yy)=v_i(g^M_i(\yy))-p^M_i(\yy)$. Now, let $i^* = \arg\min_i d_i$ (hence, $d_{i^*} \leq 1/n$) and consider the game $\calG^M_2$ where each player $i\neq i^*$ has the modified valuation function $\tilde{v}_i(x)=d_i+x$ and budget $\tilde{c}_i=d_i$, while player $i^*$ is as in $\calG^M_1$ (see Figure \ref{fig:lower-bound}). Observe that the modified utility of player $i\neq i^*$ as a function of a signal vector $\yy$ is now $\tilde{u}^M_i(\yy)=\tilde{v}_i(g^M_i(\yy))-p^M_i(\yy)=u^M_i(\yy)+d_i$. Also, since the utility of player $i \neq i^*$ is non-negative at the equilibrium $\sss$ of game $\calG_1^M$, we have that $p^M_i(\sss) \leq d_i = \tilde{c}_i$, meaning that player $i$ can also afford this payment in game $\calG^M_2$. Hence, $\sss$ is an equilibrium in $\calG^M_2$ as well (and, again, $M$ returns the same allocation $\dd$).\footnote{We remark that the conclusion $d_{i^*}\leq 1/n$ can be drawn from the assumption $\sum_{i=1}^n{d_i}\leq 1$. Hence, the proof of Theorem~\ref{thm:lower-bound} holds even for mechanisms that do not allocate the whole resource to the players. We will not consider this variation further, since no improvement of our positive results seems to be possible in this way.} 

Its liquid welfare is $\sum_i \min\{\tilde{v}_i(d_i),\tilde{c}_i\} = \min\{d_{i^*},+\infty\} + \sum_{i\neq i^*} \min\{2d_i,d_i\} = \sum_i d_i = 1$, while the optimal liquid welfare is at least $1+\sum_{i \neq i^*}d_i$, achieved at the allocation according to which the whole resource is given to player $i^*$. Hence, we conclude that the liquid price of anarchy of $M$ is $\lpoa(M) \geq \lpoa(\calG^M_2) \geq 1+\sum_{i \neq i^*}d_i = 2-d_{i^*} \geq 2-1/n$, as desired. 
\end{proof}

\begin{figure}[t!]
\centering
\includegraphics[scale=0.8]{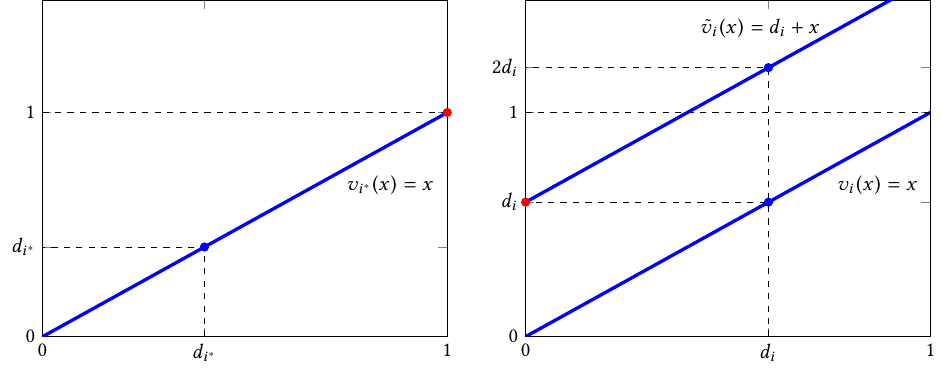}
\caption{A graphical representation of the games used in the proof of Theorem~\ref{thm:lower-bound}. The two figures depict the valuation functions of players $i^*$ and $i \neq i^*$ in games $\calG^M_1$ and $\calG^M_2$. The blue points (i.e., point $(d_{i^*},d_{i^*})$ in the left figure, and points $(d_i,d_i)$ and $(d_i,2d_i)$ in the right figure) represent the equilibrium in both games, and the optimal allocation in game $\calG^M_1$. The optimal allocation in $\calG^M_2$ is represented by the red points (i.e., point $(1,1)$ in the left figure and point $(0,d_i)$ in the right one).}
\label{fig:lower-bound}
\end{figure}

\section{The structure of worst-case games and equilibria}\label{sec:structure}
In this section, we prove our structural characterization. Given an $n$-player resource allocation mechanism $M$ (with allocation and payment functions $g^M$ and $p^M$, respectively), signal vector $\sss \in \X_n$, and an integer $j\in[n]$, define the $n$-player game $\calG^M(\sss,j)$ as follows. Every player has the affine valuation function $\tilde{v}_i(z)=\lambda^M_i(\sss)\cdot z+\kappa^M_i(\sss)$ and budget $\tilde{c}_i$, where\footnote{Note that the quantity $\left.\frac{\partial g^M_i(y,\sss_{-i})}{\partial y}\right|_{y=s_i}$ is strictly positive for every $\sss\in \X_n$.}
\begin{align*}
\lambda^M_i(\sss) &= \left(\left.\frac{\partial g^M_i(y,\sss_{-i})}{\partial y}\right|_{y=s_i}\right)^{-1}\cdot \left.\frac{\partial p^M_i(y,\sss_{-i})}{\partial y}\right|_{y=s_i}
\end{align*}
and $\kappa^M_{j}(\sss)=0$, $\tilde{c}_{j}=+\infty$, and $\kappa^M_i(\sss)=\tilde{c}_i=p^M_i(\sss)$ for every player $i\not=j$. 

Let us give an example. Let $\sss \in \X_n$ be any signal vector and consider the Kelly mechanism, which allocates to player $i$ a fraction
$g_i^\prop(\sss) = \frac{s_i}{s_i + \sum_{t \neq i} s_t}$ and requires a payment $p_i^\prop(\sss)=s_i$. Since 
\begin{align*}
\left.\frac{\partial g^\prop_i(y,\sss_{-i})}{\partial y}\right|_{y=s_i} = \frac{ \sum_{t \neq i} s_t }{ \left(s_i + \sum_{t \neq i} s_t \right)^2}
\end{align*}
and
\begin{align*}
\left.\frac{\partial p^\prop_i(y,\sss_{-i})}{\partial y}\right|_{y=s_i} = 1
\end{align*}
we have that
\begin{align*}
\lambda^\prop_i(\sss) 
= \left(\left.\frac{\partial g^M_i(y,\sss_{-i})}{\partial y}\right|_{y=s_i}\right)^{-1}\cdot \left.\frac{\partial p^M_i(y,\sss_{-i})}{\partial y}\right|_{y=s_i}
= \frac{ \left(s_i + \sum_{t \neq i} s_t \right)^2}{ \sum_{t \neq i} s_t }.
\end{align*}
Therefore, in game $\calG^\prop(\sss,j)$, player $j$ has budget $\tilde{c}_j = +\infty$ and valuation function 
$$\tilde{v}_j(z)= \frac{ \left(s_j + \sum_{t \neq j} s_t \right)^2}{ \sum_{t \neq j} s_t }\cdot z.$$
Any other player $i \neq j$ has budget $\tilde{c}_i = p_i^\prop(\sss) = s_i$ and valuation function
$$\tilde{v}_i(z)= \frac{ \left(s_i + \sum_{t \neq i} s_t \right)^2}{ \sum_{t \neq i} s_t }\cdot z + s_i.$$

In Lemma~\ref{lem:affinization}, we show that the games defined in this way are in a sense extreme in terms of the liquid price of anarchy of a mechanism $M$. Particularly for the above example, this is the one with worst liquid price of anarchy among all games induced by the $\prop$ mechanism which have the signal vector $\sss\in \X_n$ as the worst equilibrium. Then (in Lemma~\ref{lem:master}), an upper bound on the liquid price of anarchy of a mechanism is obtained by taking the supremum of the bounds implied by Lemma~\ref{lem:affinization} over all signal vectors of $\X_n$. Under mild assumptions, this bound is tight; Lemma~\ref{lem:equilibrium-CC-mechanisms} provides sufficient conditions for this.

\begin{lemma}\label{lem:affinization}
Let $\calG^M_1$ be an $n$-player resource allocation game that is induced by a mechanism $M$ with $\lpoa(\calG^M_1)>1$. Let $\sss\in \X_n$ be an equilibrium of $\calG^M_1$ of minimum liquid welfare. Then, there exists an integer $i^*\in [n]$ such that
\begin{align*}
\lpoa(\calG^M_1) \leq \frac{\LW(\tilde{\xx},\calG^M(\sss,i^*))}{\LW(g^M(\sss),\calG^M(\sss,i^*))} 
= \frac{\sum_{i\not=i^*}{p_i^M(\sss)}+\lambda^M_{i^*}(\sss)}{\sum_{i\not=i^*}{p_i^M(\sss)}+\lambda^M_{i^*}(\sss)\cdot g^M_{i^*}(\sss)},
\end{align*}
where $\tilde{\xx}=(\tilde{x}_1, ..., \tilde{x}_n)$ denotes the allocation with $\tilde{x}_{i^*}=1$ and $\tilde{x}_i = 0$ for $i\neq i^*$.
\end{lemma}

\begin{proof}
Consider an $n$-player resource allocation game $\calG_1^M$ that is induced by mechanism $M$. Let $v_i$ and $c_i$ be the valuation function and budget of player $i$, respectively. Let $\sss \in \X_n$ be the equilibrium of game $\calG_1^M$ of minimum liquid welfare. We denote by $\xx$ the optimal allocation in $\calG_1^M$. Without loss of generality, we assume that, for every player $i$, $x_i = 0$ if $v_i(0) > c_i$ and $v_i(x_i)\leq c_i$ otherwise, and we relax the allocation definition to $\sum_{i=1}^n{x_i}\leq 1$; this does not constrain the optimal liquid welfare which is $\LW(\xx,\calG_1^M) = \sum_i \min\{v_i(x_i),c_i\}$. We use $d_i=g_i^M(\sss)$ for the resource fraction allocated to player $i$ in $\sss$; let $\dd =(d_1, ..., d_n)$. 

We partition the players into the following three sets:
\begin{itemize}
\item Set $A$ consists of players $i$ with $v_i(d_i)<c_i$ and signal $s_i$ such that the derivative of their utility is equal to $0$.
\item Set $B$ consists of players $i$ with signal $s_i=0$ (hence, $d_i=0$) and negative utility derivative such that $v_i(0) < c_i$.
\item Set $\Gamma$ consists of players $i$ with signal $s_i$ such that $v_i(d_i) \geq c_i$.
\end{itemize}

First, observe that sets $A$ and $B$ cannot be both empty, since it would then be $\LW(\dd,\calG_1^M)=\sum_{i\in [n]}{c_i}\geq \LW(\xx,\calG_1^M)$, and the liquid price of anarchy of $\calG_1^M$ would be exactly $1$, contradicting the assumption of the lemma. So, in the following, we assume that at least one of $A$ and $B$ is non-empty.

Now consider the games $\calG^M(\sss,j)$ for $j\in [n]$ and let $i^* ={\arg\max}_{j\in A\cup B}{\{\lambda^M_j(\sss)\}}$. We will show that 
\begin{align}\label{eq:positive-decrease}
\LW(\dd,\calG^M_1)& \geq \LW(\dd,\calG^M(\sss,i^*))
\end{align}
and we will furthermore show that the allocation $\tilde{\xx}$ satisfies 
\begin{align}\label{eq:decrease}
\LW(\xx,\calG^M_1)-\LW(\tilde{\xx},\calG^M(\sss,i^*)) &\leq \LW(\dd,\calG^M_1)-\LW(\dd,\calG^M(\sss,i^*)).
\end{align}
In this way (recall that $\sss$ is the equilibrium of minimum liquid welfare in game $\calG^M_1$ and $\dd$ is the resulting allocation), we will have
\begin{align*}
\lpoa(\calG^M_1)&= \frac{\LW(\xx,\calG^M_1)}{\LW(\dd,\calG^M_1)} \\
&\leq  \frac{\LW(\xx,\calG^M_1)-(\LW(\xx,\calG^M_1)-\LW(\tilde{\xx},\calG^M(\sss,i^*)))}{\LW(\dd,\calG^M_1)-(\LW(\dd,\calG^M_1)-\LW(\dd,\calG^M(\sss,i^*)))}\\
&= \frac{\LW(\tilde{\xx},\calG^M(\sss,i^*))}{\LW(\dd,\calG^M(\sss,i^*))} \\
&= \frac{\sum_{i\not=i^*}{p_i^M(\sss)}+\lambda^M_{i^*}(\sss)}{\sum_{i\not=i^*}{p_i^M(\sss)}+\lambda^M_{i^*}(\sss)\cdot g^M_{i^*}(\sss)},
\end{align*}
as desired. The inequality follows by (\ref{eq:positive-decrease}) and (\ref{eq:decrease}). In particular, we remove the non-negative quantity $\LW(\dd,\calG^M_1) - \LW(\dd,\calG^M(\sss,i^*))$ from the denominator, and the smaller quantity $\LW(\xx,\calG^M_1)-\LW(\tilde{\xx},\calG^M(\sss,i^*))$ from the enumerator; observe that the latter might be negative, and hence (\ref{eq:decrease}) is not enough by itself to yield the desired inequality.
The last equality follows since all players in $\calG^M(\sss,i^*)$ besides $i^*$ have always their value capped by their budget, which is equal to their payment.

\begin{figure}[p]
\centering
\includegraphics[scale=0.65]{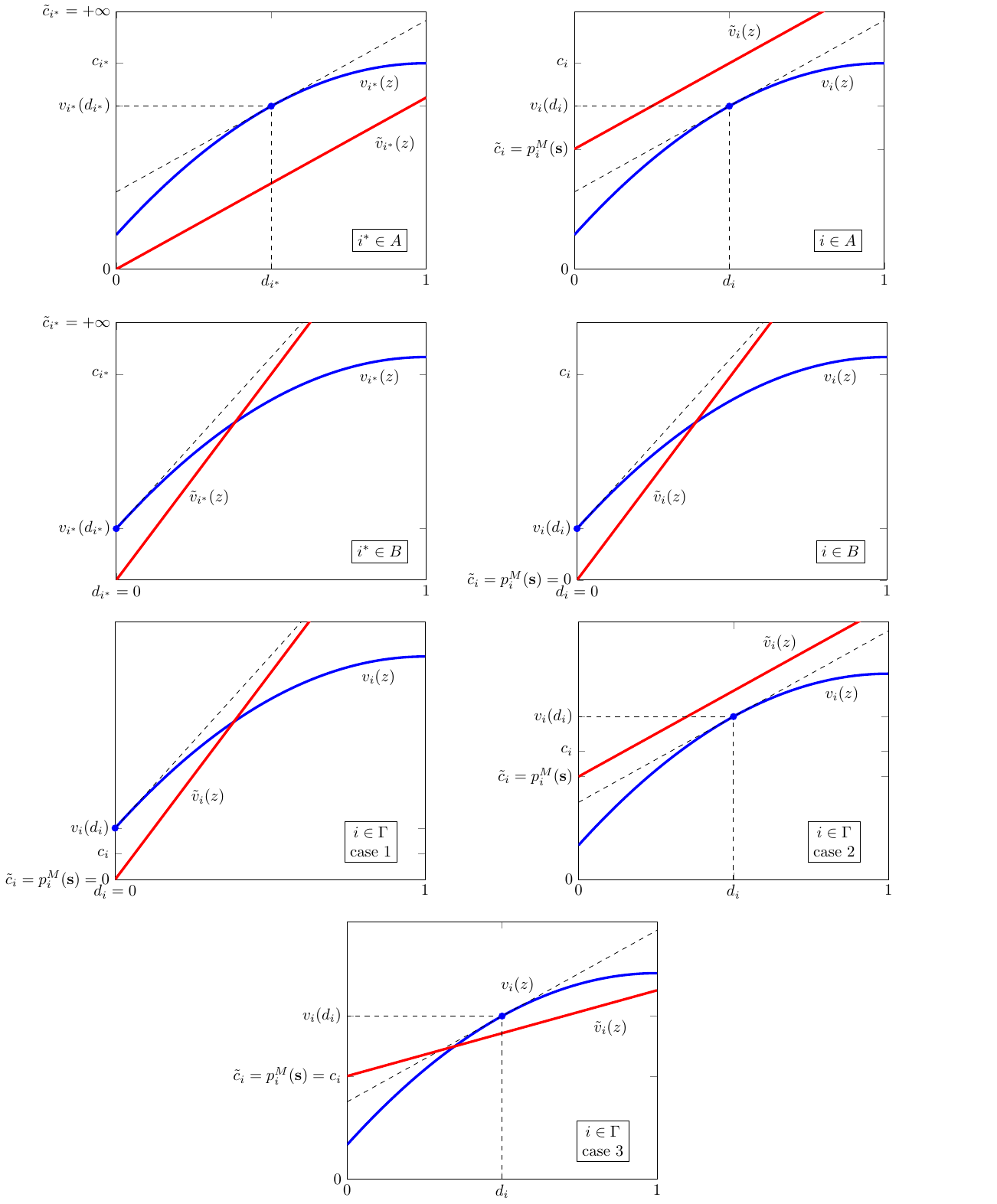}
\caption{Relation between the two games $\calG^M_1$ and $\calG^M(\sss,i^*)$ that are used in the proof of Lemma \ref{lem:affinization}; the blue lines correspond to the valuation functions of the players in $\calG^M_1$, while the red ones correspond to $\calG^M(\sss,i^*)$. The dashed line is the tangent of $v_i$ at $d_i$. The slope $\lambda_i^M(\sss)$ of the affine valuation function of player $i$ in $\calG^M(\sss,i^*)$ is greater than (third, fourth, and fifth plot), equal to (first, second, and sixth plot), or smaller than (last plot) $v_i'(d_i)$ depending on whether the utility derivative of the player is negative, zero, or positive, respectively (in particular, these are the three cases identified in the plots for $i\in \Gamma$). This follows by the definition of games $\calG^M_1$ and $\calG^M(\sss,i^*)$ and the fact that, as the utility of player $i$ in game $\calG^M_1$ has derivative 
$v'_i(d_i)\left.\frac{\partial g^M_i(y,\sss_{-i})}{\partial y}\right|_{y=s_i} - \left.\frac{\partial p^M_i(y,\sss_{-i})}{\partial y}\right|_{y=s_i}$
at equilibrium, the sign of this derivative coincides with the sign of $v'_i(d_i)-\lambda^M_i(\sss)$.} 
\label{fig:affinization}
\end{figure}

Inequality (\ref{eq:positive-decrease}) is due to the fact that the contribution of each player to the liquid welfare at $\sss$ can only decrease between the two games. Indeed, if player $i^*$ belongs to $B$, she has zero value in game $\calG^M(\sss,i^*)$. If she belongs to $A$, then her utility derivative is nullified and, hence, $v_{i^*}'(d_{i^*}) = \lambda^M_{i^*}(\sss)$. Due to the concavity of $v_{i^*}$ and since $v_{i^*}(0) \geq 0$, we get $v_{i^*}(d_{i^*}) \geq d_{i^*}v_{i^*}'(d_{i^*}) = d_{i^*}\lambda^M_{i^*}(\sss) = \tilde{v}_{i^*}(d_{i^*})$. Moreover, the contribution of player $i \neq i^*$ in $\LW(\dd,\calG^M(\sss,i^*))$ is $\tilde{c}_i = p_i^M(\sss)$ which is at most her contribution $\min\{v_i(d_i),c_i\}$ in $\LW(\dd,\calG^M_1)$ since the payment of player $i$ cannot exceed her budget in $\calG^M_1$ and her utility at equilibrium $\sss$ is non-negative. See Figure~\ref{fig:affinization} for a graphical representation of valuation functions and budgets in games $\calG^M_1$ and $\calG^M(\sss,i^*)$.

Let 
$$\delta(i) = \min\{v_i(x_i),c_i\} - \min\{\tilde{v}_i(\tilde{x}_i),\tilde{c}_i\} - \min\{v_i(d_i),c_i\} + \min\{\tilde{v}_i(d_i),\tilde{c}_i\}$$ 
denote the contribution of player $i$ to the expression 
$$\LW(\xx,\calG^M_1)-\LW(\tilde{\xx},\calG^M(\sss,i^*)) - \LW(\dd,\calG^M_1) + \LW(\dd,\calG^M(\sss,i^*)).$$ 
Then, in order to prove inequality (\ref{eq:decrease}) it suffices to prove that $\sum_i \delta(i) \leq 0$.
\begin{itemize}
\item For player $i^*$, we have that $v_{i^*}(d_{i^*}) < c_{i^*}$. Using the inequality $v_{i^*}(x_{i^*}) \leq v_{i^*}(d_{i^*}) + v_{i^*}'(d_{i^*})(x_{i^*}-d_{i^*})$ due to the concavity of the valuation function $v_{i^*}$ and the fact $\tilde{x}_{i^*}=1$, we have that
\begin{align*}
\delta(i^*) &= \min\{v_{i^*}(x_{i^*}),c_{i^*}\} - \lambda^M_{i^*}(\sss) \tilde{x}_{i^*} - v_{i^*}(d_{i^*}) + \lambda^M_{i^*}(\sss) d_{i^*} \\
&\leq v_{i^*}(x_{i^*}) - \lambda^M_{i^*}(\sss)- v_{i^*}(d_{i^*}) + \lambda^M_{i^*}(\sss) d_{i^*}\\
& \leq  v_{i^*}'(d_{i^*})(x_{i^*}-d_{i^*}) - \lambda^M_{i^*}(\sss) + \lambda^M_{i^*}(\sss) d_{i^*}.
\end{align*}
Now, we observe that (for such observations, we follow the reasoning in the caption of Figure~\ref{fig:affinization}) if player $i^*$ belongs to $A$, then $\lambda^M_{i^*}(\sss) = v_{i^*}'(d_{i^*})$, while if she belongs to $B$, then $\lambda^M_{i^*}(\sss) \geq v_{i^*}'(d_{i^*})$ and $d_{i^*} = 0$. In any case, we have that $ v_{i^*}'(d_{i^*})(x_{i^*}-d_{i^*}) \leq \lambda^M_{i^*}(\sss)(x_{i^*}-d_{i^*})$, and we obtain
\begin{align}\label{eq:delta-i-star}
\delta(i^*) &\leq \lambda^M_{i^*}(\sss)(x_{i^*}-1). 
\end{align}

\item For all players $i \neq i^*$, observe that their value is always capped by their budget in $\calG^M(\sss,i^*)$. Also, recall that $v_i(d_i) < c_i$ for every $i \in A \cup B$ and $v_i(d_i) \geq c_i$ for $i \in \Gamma$ in game $\calG_1^M$. For player $i \neq i^*$ belonging to $A$ or to $B$ we have that either $\lambda^M_i(\sss) = v_i'(d_i)$ (if $i \in A$), or $\lambda^M_i(\sss) \geq v_i'(d_i)$ and $d_i=0$ (if $i \in B$). Hence, using the concavity of $v_i$ and the fact that $\tilde{x}_i=0$, we obtain that
\begin{align}\label{eq:delta-A-Gamma-not-i}\nonumber
\delta(i) & \leq v_i(x_i) -\tilde{c}_i- v_i(d_i) + \tilde{c}_i \\\nonumber
&\leq v_i(d_i) + \lambda^M_i(\sss)(x_i - d_i) - v_i(d_i) \\
&\leq \lambda^M_{i^*}(\sss) x_i,
\end{align}
where the last inequality follows since $\lambda^M_i(\sss) \leq \lambda^M_{i^*}(\sss)$, due to the definition of player $i^*$. Otherwise, if $i \in \Gamma$, we have  
\begin{align}\label{eq:delta-Delta-Z-H}
\delta(i) = \min\{v_i(x_i),c_i\} - \tilde{c}_i - c_i + \tilde{c}_i \leq  0.
\end{align}
\end{itemize}
Hence, summing over all players, and using inequalities (\ref{eq:delta-i-star}), (\ref{eq:delta-A-Gamma-not-i}) and (\ref{eq:delta-Delta-Z-H}) as well as the fact that $\sum_i x_i \leq 1$, we obtain $\sum_i{\delta(i)} \leq 0$, and the proof is complete.
\end{proof}

We are now ready to prove the main result of this section.

\begin{lemma}\label{lem:master}
Let $M$ be an $n$-player resource allocation mechanism with allocation and payment functions $g^M$ and $p^M$, respectively. Then, its liquid price of anarchy is
\begin{align}\label{eq:master}
\lpoa(M) &\leq \sup_{\sss\in \X_n}{\left\{\frac{\sum_{i\geq 2}{p^M_i(\sss)}+\lambda^M_1(\sss)}{\sum_{i\geq 2}{p^M_i(\sss)}+\lambda^M_1(\sss)\, g^M_1(\sss)}\right\}},
\end{align}
where
\begin{align*}
\lambda^M_1(\sss) = \left(\left.\frac{\partial g_1^M(y,\sss_{-1})}{\partial y}\right|_{y=s_1}\right)^{-1} \cdot \left.\frac{\partial p_1^M(y,\sss_{-1})}{\partial y}\right|_{y=s_1}.
\end{align*}
If, in addition, for all $\sss \in \X_n$, $\sss$ is an equilibrium of game $\calG^M(\sss,1)$, then (\ref{eq:master}) holds with equality.
\end{lemma}

\begin{proof}
Using the definition of the liquid price of anarchy, Lemma \ref{lem:affinization}, and the anonymity of resource allocation mechanisms, we have
\begin{align*}
\lpoa(M) 
&= \sup_{\calG^M}{\lpoa(\calG^M)} \\
&= \sup_{\calG^M}\sup_{\sss\in \eq(\calG^M)}{\frac{\LW^*(\calG^M)}{\LW(g^M(\sss),\calG^M)}} \\
&= \sup_{\sss\in \X_n}\sup_{\calG^M: \sss\in \eq(\calG^M)}{\frac{\LW^*(\calG^M)}{\LW(g^M(\sss),\calG^M)}} \\
&\leq\sup_{\sss\in \X_n}\max_{i^*\in [n]}{\frac{\sum_{i\not=i^*}{p^M_i(\sss)}+\lambda^M_{i^*}(\sss)}{\sum_{i\not=i^*}{p^M_i(\sss)}+\lambda^M_{i^*}(\sss)\, g^M_{i^*}(\sss)}} \\
&=\sup_{\sss\in \X_n}{\frac{\sum_{i\geq 2}{p^M_i(\sss)}+\lambda^M_1(\sss)}{\sum_{i\geq 2}{p^M_i(\sss)}+\lambda^M_1(\sss)\, g^M_1(\sss)}}.
\end{align*}
Now, if $\sss \in \eq(\calG^M(\sss,1))$ for every $\sss \in \X_n$, by just considering the games $\calG^M(\sss,1)$ induced by mechanism $M$, we have 
\begin{align*}
\lpoa(M) 
&\geq \sup_{\sss\in \X_n}{\lpoa(\calG^M(\sss,1))} \\
& \geq \sup_{\sss\in \X_n}{\frac{\sum_{i\geq 2}{p^M_i(\sss)}+\lambda^M_1(\sss)}{\sum_{i\geq 2}{p^M_i(\sss)}+\lambda^M_1(\sss)\, g^M_1(\sss)}}
\end{align*}
and (\ref{eq:master}) holds with equality. The last inequality follows by comparing the liquid welfare at $\sss$ to the liquid welfare of the allocation which gives the whole resource to player $1$. Recall that all players besides player $1$ have always their value capped by their budget in game $\calG^M(\sss,1)$. 
\end{proof}

Lemma \ref{lem:master} is extremely powerful. It essentially says that no game-theoretic reasoning is needed anymore for proving upper bounds on the $\lpoa$ and, instead, all we have to do is to solve the corresponding mathematical program. Furthermore, it can be used to prove lower bounds on the $\lpoa$ without providing any explicit construction. In this case, we just need to show that the condition $\sss \in \eq(\calG^M(\sss,1))$ holds; then the tight lower bound follows by solving the same mathematical program.

Before we continue with the rest of our results, we define the class $\calC$ of mechanisms $M$ that use concave allocation functions $g^M$ and convex payment functions $p^M$. Observe that both $\prop$ and $\SH$ (as well as the $\EPYS$ mechanism presented in Section~\ref{sec:two-players}) are members of this class. With our next lemma, we prove that the condition $\sss \in \eq(\calG^M(\sss,1))$ is satisfied for any $\calC$ mechanism $M$. This will allow us to prove lower bounds in the upcoming sections. 

\begin{lemma}\label{lem:equilibrium-CC-mechanisms}
For any $n$-player resource allocation mechanism $M \in \calC$ and $\sss \in \X_n$, $\sss \in \eq(\calG^M(\sss,1))$.
\end{lemma}

\begin{proof}
Consider any mechanism $M \in \calC$ that uses a concave allocation function $g^M$ and a convex payment function $p^M$. By the definition of game $\calG^M(\sss,1)$, the utility of any player $i$, as a function of her signal $y$, is $u^M_i(y,\sss_{-i}) = \lambda^M_i(\sss)\cdot g_i^M(y,\sss_{-i})+\kappa^M_i(\sss) - p_i^M(y,\sss_{-i})$ and its derivative is
$$\frac{\partial u^M_i(y,\sss_{-i})}{\partial y} = \lambda^M_i(\sss) \frac{\partial g^M_i(y,\sss_{-i})}{\partial y} - \frac{\partial p^M_i(y,\sss_{-i})}{\partial y}.$$
Observe that, by the definition of $\lambda^M_i(\sss)$, the signal $s_i$ nullifies the utility derivative of player $i$. Furthermore, 
\begin{align*}
\frac{\partial^2 u^M_i(y,\sss_{-i})}{\partial y^2} &= \lambda^M_i(\sss) \frac{\partial^2 g^M_i(y,\sss_{-i})}{\partial y^2} - \frac{\partial^2 p^M_i(y,\sss_{-i})}{\partial y^2} 
\leq 0,
\end{align*}
for every $y$ such that $(y,\sss_{-i}) \in \X_n$ (i.e., $u_i^M(y,\sss_{-i})$ is a concave function). Hence, the signal $s_i$ actually maximizes the player's utility. 
\end{proof}

\section{Pay-your-signal mechanisms}\label{sec:pay-your-signal}  
In this section, we will exploit Lemma~\ref{lem:master} to prove tight bounds on the liquid price of anarchy of the $\prop$ and $\SH$ mechanisms. Our $\lpoa$ bounds are $2$ for $\prop$ (Theorem~\ref{thm:proportional-upper}) and $3$ for $\SH$ (Theorems~\ref{thm:SH-upper} and \ref{thm:SH-lower}). Recall that both of these mechanisms belong to class $\calC$ and, by Lemma~\ref{lem:equilibrium-CC-mechanisms}, the condition $\sss \in \eq(\calG^M(\sss,1))$ is satisfied.

\begin{theorem}\label{thm:proportional-upper}
The liquid price of anarchy of the $\prop$ mechanism is $2$. 
\end{theorem}

\begin{proof}
Consider any signal vector $\sss\in \X_n$, and let $C=\sum_{i\geq 2}{s_i}$. Since $\prop$ is a $\PYS$ mechanism, we have that $\sum_{i\geq 2}{p^{\prop}_i}(\sss)=C$ and 
$$\frac{\partial p^{\prop}_1(y,\sss_{-1})}{\partial y} =1.$$  
By the definition of the allocation function $g^{\prop}_1(y,\sss_{-1})=\frac{y}{y+C}$, we have that $$\frac{\partial g^{\prop}_1(y,\sss_{-1})}{\partial y}=\frac{C}{(y+C)^2}.$$ 
Also, since the mechanism belongs to class $\calC$, by Lemma~\ref{lem:equilibrium-CC-mechanisms}, we have that $\sss \in \eq(\calG^{\prop}(\sss,1))$. Hence, 
$$\lambda^{\prop}_1(\sss)=\frac{(s_1+C)^2}{C}$$ 
and Lemma~\ref{lem:master} yields
\begin{align*}
\lpoa(\prop) &= \sup_{s_1,C\geq 0}\frac{C+(s_1+C)^2/C}{C+(s_1+C)s_1/C} \\
&=\sup_{s_1,C\geq 0}\frac{2C^2+2s_1C+s_1^2}{C^2+s_1C+s_1^2}  \\
&= \sup_{s_1,C\geq0}{\left(2-\frac{s_1^2}{C^2+s_1C+s_1^2}\right)} \\
&=2,
\end{align*}
as desired.
\end{proof}

Notice that our proof of Theorem~\ref{thm:proportional-upper} is surprisingly short. The proof exploits Lemma~\ref{lem:master} with (\ref{eq:master}) holding with equality and, as such, it simultaneously provides a tight (upper and lower) bound. In contrast, our analysis for the SH mechanism is slightly more involved. This is mainly due to the more complicated definition of the allocation function (see Section~\ref{sec:prelim}), which requires us to distinguish between two cases, depending on whether $s_1<\max_\ell{s_\ell}$ or not. Both cases lead to inequalities that provide only an upper bound on the $\lpoa$ of SH in the proof of Theorem~\ref{thm:SH-upper}. In Theorem~\ref{thm:SH-lower}, we easily prove a matching lower bound by restricting our attention to the $2$-player version of the mechanism. Actually, the proof can be thought of as providing a tight (i.e., not only lower, but also upper) bound on the $\lpoa$ of the $2$-player version of the SH mechanism. 

\begin{theorem}\label{thm:SH-upper}
The liquid price of anarchy of the SH mechanism is at most $3$. 
\end{theorem}

\begin{proof}
We will use Lemma \ref{lem:master} and upper-bound the ratio in the RHS of (\ref{eq:master}) by $3$. Define $C=\sum_{i\geq 2}{s_i}$. First, let $\sss\in \X_n$ with $s_1<\max_\ell{s_{\ell}}$. Let $\arg\max_{\ell}{s_\ell} = i^*\not=1$. Then, by the definition of $\SH$ and the definition of $\lambda_1^{\SH}(\sss)$ in (\ref{eq:master}), we have 
\begin{align}\label{eq:lambda}
\lambda_1^{\SH}(\sss) &= \frac{s_{i^*}}{\int_0^1{\prod_{i\geq 2}{\left(1-\frac{s_i}{s_{i^*}}t\right)}\ud t}}
\end{align}
and using the Bernoulli inequality stating that $1-\gamma t \geq (1-t)^\gamma$ for $t \leq 1$ and $\gamma \in [0,1]$, (\ref{eq:lambda}) yields
\begin{align*}
\lambda_1^{\SH}(\sss) &\leq \frac{s_{i^*}}{\int_0^1{\prod_{i \geq 2}{(1-t)^{\frac{s_i}{s_{i^*}}}} \ud t}} = \frac{s_{i^*}}{\int_0^1{(1-t)^{\frac{C}{s_{i^*}}}\ud t} }
= s_{i^*}+C.
\end{align*}
Since $\SH$ is PYS, $\sum_{i\geq 2}{p_i^\SH(\sss)}=C$. Using this observation together with the last inequality, we obtain
\begin{align}\label{eq:bound-i*-is-max}
\frac{\sum_{i\geq 2}{p^{\SH}_i(\sss)}+\lambda_1^{\SH}(\sss)}{\sum_{i\geq 2}{p^{\SH}_i(\sss)}+\lambda_1^{\SH}(\sss)\,g^{\SH}_1(\sss)} &\leq \frac{2C+s_{i^*}}{C} \leq 3.
\end{align}
The inequalities follow since $\lambda_1^{\SH}(\sss)\,g^{\SH}_1(\sss)\geq 0$, $s_1\geq 0$, and $s_{i^*}\leq C$.

Now, let $\sss\in \X_n$ with $s_1=\max_\ell{s_{\ell}}$. In this case, $g^{\SH}_1(\sss)$ is defined as 
\begin{align*}
g^{\SH}_1(\sss) &= \int_0^1{\prod_{i\geq 2}{\left(1-\frac{s_i}{s_1}t\right)}\ud t}
\end{align*}
and 
\begin{align*}
\left.\frac{\partial g^{\SH}_1(y,\sss_{-1})}{\partial y}\right|_{y=s_1} &= \int_0^1{\sum_{i\geq 2}\frac{s_i}{s_1^2}t \prod_{j \neq 1,i}{\left(1-\frac{s_j}{s_1}t\right)} \ud t} \\
&\geq \sum_{i \geq 2}\frac{s_i}{s_1^2}\int_0^1 t{\prod_{j \neq 1,i}{(1-t)^{\frac{s_j}{s_1}}} \ud t} \\
&= \sum_{i\geq 2}{\frac{s_i}{s_1^2} \int_0^1{t(1-t)^{\frac{C-s_i}{s_1}} \ud t}} \\
&= \sum_{i \geq 2}{\frac{s_i}{(C-s_i + s_1)(C-s_i+ 2s_1)}} \\
&\geq \frac{C}{(C+s_1)(C+2s_1)}.
\end{align*}
Using the definition of $\lambda_1^{\SH}(\sss)$ in (\ref{eq:master}), this last inequality implies that 
\begin{align}\label{eq:lambda-2}
\lambda_1^{\SH}(\sss) &\leq \frac{(C+s_1)(C+2s_1)}{C}.
\end{align}
Also, by applying the Bernoulli inequality to the RHS of the definition of $g^{\SH}_1(\sss)$, we obtain 
\begin{align}\label{eq:g-2}
g^{\SH}_1(\sss) &\geq \int_0^1{\prod_{i\geq 2}{\left(1-t\right)^{\frac{s_i}{s_1}}}\ud t}=\int_0^1{\left(1-t\right)^{\frac{C}{s_1}}\ud t}=\frac{s_1}{C+s_1}.
\end{align}
Now, we have
\begin{align}\label{eq:bound-1-is-max}\nonumber
\frac{\sum_{i\geq 2}{p^{\SH}_i(\sss)}+\lambda_1^{\SH}(\sss)}{\sum_{i\geq 2}{p^{\SH}_i(\sss)}+\lambda_1^{\SH}(\sss)\,g^{\SH}_1(\sss)} 
&\leq \frac{C^2+(C+s_1)(C+2s_1)}{C^2+(C+s_1)(C+2s_1)\,g^{\SH}_1(\sss)} \\
&\leq \frac{2C^2+3s_1C+2s_1^2}{C^2+s_1C+2s_1^2} \leq 3.
\end{align}
The two first inequalities follow by (\ref{eq:lambda-2}) and (\ref{eq:g-2}), respectively, and the last one is obvious since $s_1,C\geq 0$.

Now, the upper bound follows by Lemma~\ref{lem:master} using (\ref{eq:bound-i*-is-max}) and (\ref{eq:bound-1-is-max}).
\end{proof}

\begin{theorem}\label{thm:SH-lower}
The liquid price of anarchy of the $\SH$ mechanism is at least $3$.
\end{theorem}

\begin{proof}
It suffices to restrict our attention to the $2$-player version of the mechanism. Let $\sss\in \X_2$ with $s_1\leq s_2$. In this case $g^{\SH}_1(\sss)=\frac{s_1}{2s_2}$ which implies that $\lambda_1^{\SH}(\sss)=2s_2$. Since the SH mechanism belongs to class $\calC$, by Lemma~\ref{lem:equilibrium-CC-mechanisms}, we have that $\sss \in \eq(\calG^{\SH}(\sss,1))$. Using Lemma \ref{lem:master}, we obtain 
\begin{align*}
\lpoa(\SH)&\geq \sup_{\sss\in \X_2: s_1\leq s_2}{\frac{3s_2}{s_2+s_1}}=3.
\end{align*}
The proof is complete.
\end{proof}

\section{Two-player mechanisms}\label{sec:two-players}
As we saw in Theorem \ref{thm:proportional-upper}, the $\prop$ mechanism has an $\lpoa$ of exactly $2$ even in the case of two players. In contrast, our lower bound of $3/2$ for $2$-player mechanisms in Theorem~\ref{thm:lower-bound} seems to leave room for improvements. Such improvements are indeed possible as we show with the mechanisms that we present in this section. Interestingly, the $\EPYS$ mechanism that is defined in the following is also proved to have optimal $\lpoa$ among all $2$-player $\PYS$ mechanisms with concave allocation functions.

\subsection{The $\EPYS$ mechanism}\label{sec:EPYS}
Let $\beta\approx 1.792$ be the solution of the equation
$\frac{1}{\beta}-\frac{1}{\beta}\exp\left(-\frac{\beta}{\beta-1}\right) = \frac{1}{2}$
and define mechanism $\EPYS$ to be the $\PYS$ $2$-player mechanism that uses the allocation function  
\begin{align*}
g^{\EPYS}_i(\sss) &= 
\begin{cases} 
\frac{1}{\beta}-\frac{1}{\beta}\exp\left(-\frac{\beta}{\beta-1}\cdot \frac{s_i}{s_{3-i}}\right) & s_i\leq s_{3-i}\\
\frac{\beta-1}{\beta}+\frac{1}{\beta}\exp\left(-\frac{\beta}{\beta-1}\cdot \frac{s_{3-i}}{s_i}\right) & s_i> s_{3-i}
\end{cases}
\end{align*}
for player $i\in \{1,2\}$ and (non-zero) signal vector $\sss=(s_1,s_2)$. Due to the definition of $\beta$, $\EPYS$ is a well-defined resource allocation mechanism: it is anonymous, with an increasing and differentiable allocation function, which allocates the whole resource when some player has non-zero signal. Moreover, $\EPYS$ belongs to class $\calC$: the allocation function can be seen to be concave (see also Figure~\ref{fig:allocation-functions}) and the payment function is, of course, convex. The $\lpoa$ bound statement for $\EPYS$ follows.

\begin{theorem}\label{thm:E2-PYS-bound}
The liquid price of anarchy of the $\EPYS$ mechanism is $\beta \approx 1.792$.
\end{theorem}

\begin{proof}
We will prove the theorem using Lemma~\ref{lem:master}. Let $\sss\in \X_2$. Due to Lemma~\ref{lem:equilibrium-CC-mechanisms}, we have that $\sss \in \eq(\calG^{\EPYS}(\sss,1))$. Since $\EPYS$ is a PYS mechanism, we have that $p_1^\EPYS(\sss) = s_1$, which yields 
\begin{align*}
\left.\frac{\partial p_1^{\EPYS}(y,s_2)}{\partial y}\right|_{y=s_1} = 1.
\end{align*}
Next, we distinguish between two cases. First, assume that $s_1\leq s_2$; in this case, the allocation of player $1$ is 
\begin{align*}
g^{\EPYS}_1(s_1,s_2) &= \frac{1}{\beta}-\frac{1}{\beta} \exp\left(-\frac{\beta}{\beta-1}\cdot \frac{s_1}{s_2}\right)
\end{align*}
and, thus, the derivative is equal to
\begin{align*}
\left.\frac{\partial g_1^{\EPYS}(y,s_2)}{\partial y}\right|_{y=s_1} &= \frac{1}{(\beta-1)s_2} \exp\left(-\frac{\beta}{\beta-1}\cdot\frac{s_1}{s_2}\right).
\end{align*}
Therefore, $\lambda^{\EPYS}_1(\sss)$ is defined as 
\begin{align*}
\lambda^{\EPYS}_1(\sss)
&=\left(\left.\frac{\partial g^\EPYS_i(y,\sss_{-i})}{\partial y}\right|_{y=s_i}\right)^{-1}\cdot \left.\frac{\partial p^\EPYS_i(y,\sss_{-i})}{\partial y}\right|_{y=s_i} 
=(\beta-1)s_2 \exp\left(\frac{\beta}{\beta-1}\cdot\frac{s_1}{s_2}\right).
\end{align*}
By substituting $p^{\EPYS}_2(\sss)$, $\lambda^{\EPYS}_1(\sss)$, and $g^{\EPYS}_1(\sss)$ in (\ref{eq:master}), we obtain 
\begin{align}\label{eq:e2-pys-first-case}\nonumber
&\frac{p^{\EPYS}_2(\sss)+\lambda^{\EPYS}_1(\sss)}{p^{\EPYS}_2(\sss)+\lambda^{\EPYS}_1(\sss)\,g^{\EPYS}_1(\sss)} \\
&=  \frac{s_2 + (\beta-1)s_2 \exp\left(\frac{\beta}{\beta-1}\cdot\frac{s_1}{s_2}\right) }{s_2 + \frac{\beta-1}{\beta}s_2 \exp\left(\frac{\beta}{\beta-1}\cdot\frac{s_1}{s_2}\right) \left(1 -  \exp\left(-\frac{\beta}{\beta-1}\cdot \frac{s_1}{s_2}\right)\right) } 
= \beta.
\end{align}

For the second case where $s_1>s_2$, we have that 
\begin{align*}
g^{\EPYS}_1(s_1,s_2) &= \frac{\beta-1}{\beta}+\frac{1}{\beta} \exp\left(-\frac{\beta}{\beta-1}\cdot \frac{s_2}{s_1}\right)
\end{align*}
and
\begin{align*}
\left.\frac{\partial g_1^{\EPYS}(y,s_2)}{\partial y}\right|_{y=s_1} &= \frac{s_2}{(\beta-1)s_1^2} \exp\left(-\frac{\beta}{\beta-1}\cdot\frac{s_2}{s_1}\right).
\end{align*}
Now, it is
\begin{align*}
\lambda^{\EPYS}_1(\sss)
&=\left(\left.\frac{\partial g^\EPYS_i(y,\sss_{-i})}{\partial y}\right|_{y=s_i}\right)^{-1}\cdot \left.\frac{\partial p^\EPYS_i(y,\sss_{-i})}{\partial y}\right|_{y=s_i} \\
&=\frac{(\beta-1)s_1^2}{s_2} \exp\left(\frac{\beta}{\beta-1}\cdot\frac{s_2}{s_1}\right)
\end{align*}
By substituting $p^{\EPYS}_2(\sss)$, $\lambda^{\EPYS}_1(\sss)$, and $g^{\EPYS}_1(\sss)$ in (\ref{eq:master}), we obtain  
\begin{align}\label{eq:e2-pys-second-case}\nonumber
&\frac{p^{\EPYS}_2(\sss)+\lambda^{\EPYS}_1(\sss)}{p^{\EPYS}_2(\sss)+\lambda^{\EPYS}_1(\sss)\,g^{\EPYS}_1(\sss)} \\\nonumber
&= \frac{ s_2 +  \frac{(\beta-1)s_1^2}{s_2} \exp\left(\frac{\beta}{\beta-1}\cdot\frac{s_2}{s_1}\right) }{ s_2 + \frac{(\beta-1)s_1^2}{s_2} \exp\left(\frac{\beta}{\beta-1}\cdot\frac{s_2}{s_1}\right) 
\cdot \left( \frac{\beta-1}{\beta}+\frac{1}{\beta}\exp\left(-\frac{\beta}{\beta-1}\cdot \frac{s_2}{s_1}\right) \right) } \\\nonumber 
&= \frac{ s_2 +  \frac{(\beta-1)s_1^2}{s_2} \exp\left(\frac{\beta}{\beta-1}\cdot\frac{s_2}{s_1}\right) }
{ s_2 + \frac{(\beta-1)^2}{\beta}\frac{s_1^2}{s_2} \exp\left(\frac{\beta}{\beta-1}\cdot\frac{s_2}{s_1}\right) 
+\frac{\beta-1}{\beta} \frac{s_1^2}{s_2} } \\
&= \beta \frac{1+(\beta-1)\left(\frac{s_1}{s_2}\right)^2\exp\left(\frac{\beta}{\beta-1}\cdot\frac{s_2}{s_1}\right)}{\beta+(\beta-1)^2\left(\frac{s_1}{s_2}\right)^2\exp\left(\frac{\beta}{\beta-1}\cdot\frac{s_2}{s_1}\right)+(\beta-1)\left(\frac{s_1}{s_2}\right)^2}
\leq \beta.
\end{align}
The inequality follows since the quantity at its left is decreasing in $s_1/s_2$ (its derivative with respect to $s_1/s_2$ can be shown by tedious calculations to be non-positive for $s_1/s_2\geq 1$) and, hence, it is upper-bounded by its value for $s_1/s_2=1$; this is equal to $\beta$ by its definition.

The theorem follows by Lemma~\ref{lem:master} using (\ref{eq:e2-pys-first-case}) and (\ref{eq:e2-pys-second-case}).
\end{proof}

We remark that a preliminary analysis similar to the first half of the proof of Theorem \ref{thm:E2-PYS-bound} inspired the design of the $\EPYS$ mechanism (as well as that of the $\SR$ mechanism that is defined later) at first place. By keeping the allocation function as the unknown and requiring that the RHS of (\ref{eq:master}) is equal to some value $\alpha$ for all signal vectors $\sss\in \X_2$ with $s_1\leq s_2$ (this is essentially what (\ref{eq:e2-pys-first-case}) captures), we obtained a first-order differential equation which, using the appropriate conditions so that the resulting mechanism is valid, led to $\EPYS$ (for $\alpha=\beta$). Luckily, for signal vectors $\sss\in \X_2$ with $s_1>s_2$, we were able to show that the RHS of (\ref{eq:master}) is at most $\alpha$; see inequality (\ref{eq:e2-pys-second-case}). 

We now show that $\EPYS$ has optimal $\lpoa$ among $2$-player $\PYS$ mechanisms in class $\calC$. The proof makes use of Lemma~\ref{lem:master} and a simple differential inequality that involves the allocation function. 

\begin{theorem}\label{thm:PYS-optimality}
Any $2$-player $\PYS$ mechanism with a concave allocation function has liquid price of anarchy at least $\beta \approx 1.792$.
\end{theorem}

\begin{proof}
For the sake of contradiction, assume that there exists a $\PYS$ mechanism $M$ that has liquid price of anarchy $\beta'<\beta$. Denote by $f:\R_{\geq0}\rightarrow [0,1]$ the function defined as $f(y)=g^M_1(y,1)$. Then, by applying Lemma~\ref{lem:master} with $\sss = (y,1) \in \X_2$ to $M$ we have $\lambda^M_1(y,1)=1/f'(y)$ and $\lpoa(M)\geq \frac{1+1/f'(y)}{1+f(y)/f'(y)}$ for every $y\in [0,1]$. By our assumption $\lpoa(M)\leq \beta'$, we get the differential inequality
\begin{align*}
(\beta'-1)f'(y)+\beta' f(y)\geq 1
\end{align*}
for every $y\in [0,1]$. Using Gr\"onwall's inequality, $f(y)$ is lower-bounded by the solution of the corresponding differential equation. Due to the condition $f(0)=0$, this yields
\begin{align*}
f(y) &\geq \frac{1}{\beta'}-\frac{1}{\beta'}\exp\left(-\frac{\beta'}{\beta'-1} y\right)
\end{align*} 
and, hence,
\begin{align*}
\frac{1}{2}=f(1) &\geq \frac{1}{\beta'}-\frac{1}{\beta'}\exp\left(-\frac{\beta'}{\beta'-1} \right)> \frac{1}{\beta}-\frac{1}{\beta}\exp\left(-\frac{\beta}{\beta-1} \right),
\end{align*}
which contradicts the definition of $\beta$. The last inequality follows since the function $\frac{1}{z}-\frac{1}{z}\exp\left(-\frac{z}{z-1}\right)$ is decreasing in the interval $[1,2]$.
\end{proof}

\subsection{The $\SR$ mechanism}\label{sec:SR}  
Let us now define a non-$\PYS$ mechanism that has considerably better $\lpoa$ than $\EPYS$ and almost matches the lower bound of $3/2$ from Theorem~\ref{thm:lower-bound} for $2$-player mechanisms. Let
$\gamma\approx 1.529$ be the solution of the equation
$\frac{1}{\gamma}-\frac{1}{\gamma}\exp\left(-\frac{\gamma}{2(\gamma-1)}\right) = \frac{1}{2}$ and define mechanism $\SR$ to be the $2$-player mechanism that uses the allocation function (see Figure \ref{fig:allocation-functions} for a comparison of the allocation functions of $\prop$, $\SH$, $\EPYS$, and $\SR$)
\begin{align*}
g^{\SR}_i(\sss) &= 
\begin{cases} 
\frac{1}{\gamma}-\frac{1}{\gamma}\exp\left(-\frac{\gamma}{2(\gamma-1)}\cdot \left(\frac{s_i}{s_{3-i}}\right)^2\right) & s_i\leq s_{3-i}\\
\frac{\gamma-1}{\gamma}+\frac{1}{\gamma}\exp\left(-\frac{\gamma}{2(\gamma-1)}\cdot \left(\frac{s_{3-i}}{s_i}\right)^2\right) & s_i> s_{3-i}
\end{cases}
\end{align*}
and the payment function $p^{\SR}_i(\sss)=s_i/s_{3-i}$ for player $i\in \{1,2\}$ and (non-zero) signal vector $\sss=(s_1,s_2)$. By the general conventions of Section~\ref{sec:prelim}, the payments are $0$ when some of the signals are equal to zero. Due to the definition of $\gamma$, $\SR$ is a well-defined resource allocation mechanism. However, observe that $\SR$ does not belong to class $\calC$ (the allocation function is not concave; see Figure~\ref{fig:allocation-functions}) and the condition $\sss \in \eq(\calG^{\SR}(\sss,1))$ is not guaranteed to be satisfied. Next, we will prove an upper bound on the $\lpoa$ of $\SR$. The proof follows in a similar way to the proof of Theorem \ref{thm:E2-PYS-bound}, but it does not provide a tight bound.

\begin{figure}[t!]
\centering
\includegraphics[scale=0.8]{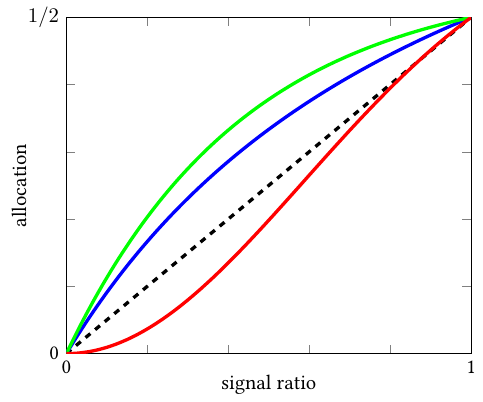}
\caption{A comparison of the allocation function $g^M_i$ used by $\EPYS$ (in green), (the $2$-player version of) $\prop$ (in blue), $\SH$ (dashed), and $\SR$ (in red) as a function of $s_i/s_{3-i}$ for $s_i\leq s_{3-i}$. Among these mechanisms, $\SR$ is the only one with a non-concave allocation function.}
\label{fig:allocation-functions}
\end{figure}

\begin{theorem}\label{thm:RS-bound}
The liquid price of anarchy of the $\SR$ mechanism is at most $\gamma \approx 1.529$.
\end{theorem}

\begin{proof}
We will prove the theorem by mimicking the proof of Theorem \ref{thm:E2-PYS-bound}. Let $\sss\in \X_2$. Again, we distinguish between two cases. First, assume that $s_1\leq s_2$. Then, since
\begin{align*}
g^{\SR}_1(s_1,s_2) &= \frac{1}{\gamma}-\frac{1}{\gamma}\exp\left(-\frac{\gamma}{2(\gamma-1)}\cdot\left(\frac{s_1}{s_2}\right)^2\right),
\end{align*}
the derivative of the allocation for player $1$ is
\begin{align*}
\left.\frac{\partial g_1^{\SR}(y,s_2)}{\partial y}\right|_{y=s_1} &= \frac{s_1}{(\gamma-1)s_2^2} \exp\left(-\frac{\gamma}{2(\gamma-1)}\cdot\left(\frac{s_1}{s_2}\right)^2\right).
\end{align*}
Also, since the payment function used by $\SR$ is equal to the signal ratio, its derivative for player $1$ is 
\begin{align*}
\left.\frac{\partial p_1^{\SR}(y,s_2)}{\partial y}\right|_{y=s_1} &= \frac{1}{s_2}.
\end{align*}
Therefore,
\begin{align*}
\lambda^{\SR}_1(\sss) 
&= \left(\left.\frac{\partial g^\SR_i(y,\sss_{-i})}{\partial y}\right|_{y=s_i}\right)^{-1}\cdot \left.\frac{\partial p^\SR_i(y,\sss_{-i})}{\partial y}\right|_{y=s_i} \\
&= (\gamma-1)\frac{s_2}{s_1} \exp\left(\frac{\gamma}{2(\gamma-1)}\cdot\left(\frac{s_1}{s_2}\right)^2\right).
\end{align*}
By substituting $p^{\SR}_2(\sss)$, $\lambda^{\SR}_1(\sss)$, and $g^{\SR}_1(\sss)$ to (\ref{eq:master}), we can now easily verify that
\begin{align}\label{eq:e2-sr-first-case}
\frac{p^{\SR}_2(\sss)+\lambda^{\SR}_1(\sss)}{p^{\SR}_2(\sss)+\lambda^{\SR}_1(\sss)\,g^{\SR}_1(\sss)}&=\gamma.
\end{align}

For the second case where $s_1>s_2$, the allocation is defined as 
\begin{align*}
g^{\SR}_1(s_1,s_2) &= \frac{\gamma-1}{\gamma}+\frac{1}{\gamma} \exp\left(-\frac{\gamma}{2(\gamma-1)}\cdot \left(\frac{s_2}{s_1}\right)^2\right)
\end{align*}
and the derivative is
\begin{align*}
\left.\frac{\partial g_1^{\SR}(y,s_2)}{\partial y}\right|_{y=s_1} &= \frac{s_2^2}{(\gamma-1)s_1^3} \exp\left(-\frac{\gamma}{2(\gamma-1)}\cdot\left(\frac{s_2}{s_1}\right)^2\right).
\end{align*}
The payment derivative is again equal to $1/s_2$ and, hence, 
\begin{align*}
\lambda^{\SR}_1(\sss)
&= \left(\left.\frac{\partial g^\SR_i(y,\sss_{-i})}{\partial y}\right|_{y=s_i}\right)^{-1}\cdot \left.\frac{\partial p^\SR_i(y,\sss_{-i})}{\partial y}\right|_{y=s_i} \\
&=(\gamma-1)\left(\frac{s_1}{s_2}\right)^3 \exp\left(\frac{\gamma}{2(\gamma-1)}\cdot\left(\frac{s_2}{s_1}\right)^2\right).
\end{align*}
By substituting $p^{\SR}_2(\sss)$, $\lambda^{\SR}_1(\sss)$, and $g^{\SR}_1(\sss)$, we obtain 
\begin{align}\label{eq:e2-sr-second-case}\nonumber
&\frac{p^{\SR}_2(\sss)+\lambda^{\SR}_1(\sss)}{p^{\SR}_2(\sss)+\lambda^{\SR}_1(\sss)\,g^{\SR}_1(\sss)} \\\nonumber
&= \frac{ \frac{s_2}{s_1} + (\gamma-1)\left(\frac{s_1}{s_2}\right)^3 \exp\left(\frac{\gamma}{2(\gamma-1)}\cdot\left(\frac{s_2}{s_1}\right)^2\right) }{ \frac{s_2}{s_1} + (\gamma-1)\left(\frac{s_1}{s_2}\right)^3 \exp\left(\frac{\gamma}{2(\gamma-1)}\cdot\left(\frac{s_2}{s_1}\right)^2\right) \cdot \left( \frac{\gamma-1}{\gamma}+\frac{1}{\gamma} \exp\left(-\frac{\gamma}{2(\gamma-1)}\cdot \left(\frac{s_2}{s_1}\right)^2\right) \right) } \\
&= \gamma \frac{1+(\gamma-1)\left(\frac{s_1}{s_2}\right)^4\exp\left(\frac{\gamma}{2(\gamma-1)}\cdot\left(\frac{s_2}{s_1}\right)^2\right)}{\gamma+(\gamma-1)^2\left(\frac{s_1}{s_2}\right)^4\exp\left(\frac{\gamma}{2(\gamma-1)}\cdot\left(\frac{s_2}{s_1}\right)^2\right)+(\gamma-1)\left(\frac{s_1}{s_2}\right)^4}
\leq \gamma.
\end{align}
The inequality follows because the quantity at its left is decreasing in $s_1/s_2$ (its derivative with respect to $s_1/s_2$ can be shown by tedious calculations to be non-positive for $s_1/s_2\geq 1$) and, hence, it is upper-bounded by its value for $s_1/s_2=1$; this is equal to $\gamma$ by its definition.

The theorem follows by Lemma~\ref{lem:master} using (\ref{eq:e2-sr-first-case}) and (\ref{eq:e2-sr-second-case}).
\end{proof}

\subsection{The $\SHSR$ mechanism}
Interestingly, we now show that there exists a simple $2$-player mechanism that is part of class $\calC$ and achieves a tight $\lpoa$ equal to the golden ratio $\phi \approx 1.618$, improving upon the bound of $1.792$ achieved by $\EPYS$ for mechanisms in $\calC$. This mechanism uses the allocation function of $\SH$ and the signal-ratio payment function of $\SR$. In the following, we will refer to this $2$-player mechanism as $\SHSR$.

\begin{theorem}\label{thm:SHSR-bound}
The liquid price of anarchy of the $\SHSR$ mechanism is $\phi \approx 1.618$.
\end{theorem}

\begin{proof}
Again, we will prove the theorem using Lemma~\ref{lem:master}. Let $\sss\in \X_2$. Due to Lemma~\ref{lem:equilibrium-CC-mechanisms}, since the allocation function is concave and the payment function is convex, we have that $\sss \in \eq(\calG^{\SHSR}(\sss,1))$. By the definition of the payment function, we have that $p_1^\SHSR(\sss) = \frac{s_1}{s_2}$, which yields that
\begin{align*}
\left.\frac{\partial p_1^{\SHSR}(y,s_2)}{\partial y}\right|_{y=s_1} = \frac{1}{s_2}.
\end{align*}
Next, we distinguish between two cases. First, assume that $s_1\leq s_2$. Then, the allocation of player $1$ is $g_1^\SHSR(s_1,s_2) = \frac{s_1}{2s_2}$, which yields that
\begin{align*}
\left.\frac{\partial g_1^{\SHSR}(y,s_2)}{\partial y}\right|_{y=s_1} &= \frac{1}{2s_2}.
\end{align*}
Therefore, $\lambda^{\SHSR}_1(\sss) = 2$. By setting $z = \frac{s_1}{s_2}$, and substituting $p^{\SHSR}_2(\sss)$, $\lambda^{\SHSR}_1(\sss)$, and $g^{\SHSR}_1(\sss)$ in (\ref{eq:master}), we obtain 
\begin{align}\label{eq:SHSR-first-case}
\frac{p^{\SHSR}_2(\sss)+\lambda^{\SHSR}_1(\sss)}{p^{\SHSR}_2(\sss)+\lambda^{\SHSR}_1(\sss)\,g^{\SHSR}_1(\sss)} &= \frac{s_2^2 + 2s_1s_2}{s_2^2 + s_1^2} = \frac{1+2z}{1+z^2} \leq \phi.
\end{align}
The inequality follows since the maximum value of the expression $\frac{1+2z}{1+z^2}$ is $\phi$, which occurs at $z=\phi-1$.

For the second case where $s_1>s_2$, since $g_1(s_1,s_2) = 1-\frac{s_2}{2s_1}$, we have that  
\begin{align*}
\left.\frac{\partial g_1^{\SHSR}(y,s_2)}{\partial y}\right|_{y=s_1} &= \frac{s_2}{2s_1^2}
\end{align*}
and, therefore, $\lambda^{\SHSR}_1(\sss)=2\left(\frac{s_1}{s_2}\right)^2$.
By setting $z = \frac{s_2}{s_1}$, and substituting $p^{\SHSR}_2(\sss)$, $\lambda^{\SHSR}_1(\sss)$, and $g^{\SHSR}_1(\sss)$ in (\ref{eq:master}), we obtain  
\begin{align}\label{eq:SHSR-second-case}\nonumber
\frac{p^{\SHSR}_2(\sss)+\lambda^{\SHSR}_1(\sss)}{p^{\SHSR}_2(\sss)+\lambda^{\SHSR}_1(\sss)\,g^{\SHSR}_1(\sss)} &= \frac{ \frac{s_2}{s_1} + 2 \left(\frac{s_1}{s_2}\right)^2 }{ \frac{s_2}{s_1} + 2 \left(\frac{s_1}{s_2}\right)^2 \left( 1- \frac{s_2}{2s_1}\right) } \\
&= 1 + \frac{z}{z^3-z+2} \leq \phi.
\end{align}
Here, the inequality follows since the maximum value of the expression $\frac{z}{z^3-z+2}$ is $1/2$, which occurs for $z=1$. 

The theorem follows by Lemma~\ref{lem:master} using (\ref{eq:SHSR-first-case}) and (\ref{eq:SHSR-second-case}).
\end{proof}

\section{Open problems and possible extensions}\label{sec:extensions}
Even though we have revealed an almost complete picture on the liquid price of anarchy of resource allocation mechanisms, our work leaves an interesting open problem: is the $2-1/n$ bound achievable (as opposed to $2$), preferably by a simple mechanism? In particular, is there a mechanism with a proportional allocation function and appropriate non-PYS payments that achieves this $\lpoa$ bound? This question seems technically challenging even for the case of two players only. 

Regarding the liquid price of anarchy over more general equilibrium concepts (e.g., correlated equilibria) or settings with incomplete information (and Bayes-Nash equilibria), is the Kelly mechanism still optimal within low-order terms? We are far from answering this question. The papers by \citet{CV16} and by \citet{CST16} present $\lpoa$ bounds for Kelly, but these are not known to be tight. We conjecture that the proof of tight $\lpoa$ bounds over more general equilibrium concepts for any resource allocation mechanism should exploit the structure of worst-case games and equilibria as we did in the current paper for pure Nash equilibria. Unfortunately, extending our characterization from Section \ref{sec:structure} to more general equilibrium concepts seems elusive at this point.

There are several interesting extensions of our setting that could be considered. {\em Budget-aware} mechanisms, which have access to the budget value of each player, constitute a first such extension. Of course, our analysis for mechanisms $\prop$, $\SH$, $\EPYS$, and $\SR$ carries over to this case. In contrast, our lower bound (Theorem~\ref{thm:lower-bound}) is not true anymore. The proof constructs two games, in which almost every player has different budgets. The main property we have exploited in that proof (for non-budget-aware mechanisms) is that the strategic behavior of the players results in the same set of equilibria in both games. This argument fails for budget-aware mechanisms; a small change in the budget of a single player could be enough to alter the set of equilibria. So, in principle, one might hope even for full efficiency at equilibria (i.e., $\lpoa$ equal to $1$) in this case, analogously to the results of \citet{MB06}, \citet{JT09}, and \citet{YH07} in the no-budget setting. Interestingly, our next statement rules out this possibility. 

\begin{theorem}\label{thm:lower-bound-known}
For $n\geq 2$, every $n$-player budget-aware resource allocation mechanism has liquid price of anarchy at least $4/3$.
\end{theorem}

\begin{proof}
Let $M$ be any $n$-player budget-aware resource allocation mechanism that uses an allocation function $g^M$ and a payment function $p^M$. Let $\sss=(s_1, ..., s_n)$ be an equilibrium of the game $\calG^M_1$ induced by $M$ for players with valuations $v_i(x)=x$ for $i \in \{1,2\}$ and $v_i(x) = 0$ for $i\geq 3$, and budgets $c_i=1$ for every $i \in [n]$. Assume that the allocation returned by $M$ at this equilibrium is $\dd = (d_1, ..., d_n)$. Without loss of generality, we may assume that one of the first two players (say, player $1$) gets a resource share of at most $1/2$.

Recall that, for every signal vector $\yy$, the utility of any player $i$ is defined as $u^M_i(\yy)=v_i(g^M_i(\yy))-p^M_i(\yy)$. Now, consider the game $\calG^M_2$ where player $2$ has the modified valuation function $\tilde{v}_2(x)=1+x$ while all other players are as in $\calG^M_1$; the budgets are the same in both games and are known to the mechanism. Observe that the modified utility of player $2$ is now $\tilde{u}^M_2(\yy)=\tilde{v}_2(g^M_2(\yy))-p^M_2(\yy)=u^M_2(\yy)+1$. Hence, $\sss$ is an equilibrium in $\calG^M_2$ as well and $M$ returns the same allocation $\dd$ again. 

Clearly, due to the definition of the valuation functions, the contribution of players $i\geq 3$ in the liquid welfare (in any state of the game) is zero. Hence, the liquid welfare at equilibrium is $\min\{\tilde{v}_1(d_1),c_1\} + \min\{\tilde{v}_2(d_2),c_2\} = d_1 + 1 \leq 3/2$, while the optimal liquid welfare is equal to $2$, achieved at the allocation according to which the whole resource is given to player $1$. We conclude that the liquid price of anarchy of $M$ is $\lpoa(M) \geq \lpoa(\calG^M_2) \geq 4/3$, as desired. 
\end{proof}

In spite of the lower bound in Theorem \ref{thm:lower-bound-known}, whether budget-aware resource allocation mechanisms can have an $\lpoa$ better than $2-1/n$ is an important open problem. This seems to be a technically non-trivial and extremely challenging task though.

Another possible extension of our setting could be to allow the players to declare their budgets to the mechanism in addition to their scalar signal. Taking this approach to its extreme, one could imagine resource allocation mechanisms which ask the players to submit {\em multi-dimensional} signals. At first glance, this seems to lead to much more powerful mechanisms than the ones we have considered here (since a signal consisting of many scalars could allow a player to encode her valuation function and budget). Surprisingly, this {\em higher level of expressiveness}~\citep{DFP18} has no consequences to the $\lpoa$ at all and our lower bound of $2-1/n$ captures such mechanisms as well. Indeed, by inspecting the two games used in the proof of Theorem \ref{thm:lower-bound}, we can verify that the same signal vector (no matter whether signals are single- or multi-dimensional) leads to the same allocation by the mechanism and the same strategic behavior of the players in both games. This observation applies to the proof of Theorem \ref{thm:lower-bound-known} as well.

Finally, we believe that the liquid welfare is an appropriate efficiency benchmark for auctions with budget-constrained players. The recent paper by \citet{AFGR15} studies the $\lpoa$ of simultaneous first-price auctions; obtaining similar results for other auction formats (e.g., see the recent survey of \citet{RST16}) is certainly important. Needless to say, we do not expect that the liquid welfare is unique as a measure of efficiency in settings with budgets. Defining alternative efficiency benchmarks and studying the price of anarchy with respect to them would shed extra light to the strengths and weaknesses of auction mechanisms.

\section*{Acknowledgments}
This work has been partially supported by Caratheodory research grant E.114 from the University of Patras, by a PhD scholarship from the Onassis Foundation, and by the European Research Council (ERC) under grant number 639945 (ACCORD).

% Bibliography
\bibliographystyle{named}
\bibliography{resource}

\end{document}